\documentclass[conference]{IEEEtran}
\IEEEoverridecommandlockouts

\usepackage{cite}
\usepackage{amsmath,amssymb,amsfonts}
\usepackage{algorithmic}
\usepackage{graphicx}
\usepackage{textcomp}
\usepackage{xcolor}
\usepackage{amssymb}
\usepackage{multirow}
\usepackage[table]{xcolor}
\usepackage{cite}
\usepackage{amsthm}
\newtheorem{theorem}{Theorem}[section]
\newtheorem{lemma}[theorem]{Lemma}
\def\BibTeX{{\rm B\kern-.05em{\sc i\kern-.025em b}\kern-.08em
    T\kern-.1667em\lower.7ex\hbox{E}\kern-.125emX}}
\begin{document}

\newcommand{\squishlist}{
    \begin{list}{$\bullet$}
        { \setlength{\itemsep}{0pt}      \setlength{\parsep}{0pt}
            \setlength{\topsep}{0.5pt}       \setlength{\partopsep}{0pt}
            \setlength{\listparindent}{-2pt}
            \setlength{\itemindent}{-5pt}
            \setlength{\leftmargin}{0.5em} \setlength{\labelwidth}{0em}
            \setlength{\labelsep}{0.2em} } }

\newcommand{\squishend}{
\end{list}  }
\newcommand\rebutt[1]{\textcolor{blue}{#1}}
\newtheorem{defi}{Definition}
\newcommand{\jl}[1]{{\textsf{{\color{purple}{Ji: }   #1 }}}}

\title{Non-Clifford Fusion: T-Gate Optimization for Quantum Simulation}

\author{\IEEEauthorblockN{Yingheng Li\IEEEauthorrefmark{1}$^,$\IEEEauthorrefmark{2}, Xulong Tang\IEEEauthorrefmark{2}, Paul Hovland\IEEEauthorrefmark{1}, Ji Liu\IEEEauthorrefmark{1}}
\IEEEauthorblockA{\IEEEauthorrefmark{1} Mathematics and Computer Science Division, Argonne National Laboratory}
\IEEEauthorblockA{\IEEEauthorrefmark{2} Department of Computer Science, University of Pittsburgh}
}

\maketitle

\begin{abstract}
Hamiltonian simulation is a key quantum algorithm for modeling complex systems. To implement a Hamiltonian simulation, it is typically decomposed into a list of Pauli strings, each corresponds to a $R_Z$ rotation gate with many Clifford gates. These RZ gates are generally synthesized into a sequence of Clifford and T gates in fault-tolerant quantum computers, where the T-gate count and T-gate depth are critical metrics for such systems. In this paper, we propose NCF, a compilation framework that reduces both the T-gate count and T-gate depth for Hamiltonian simulation. NCF partitions Pauli strings into groups, where each group can be conjugated (i.e., transformed) into a list of Pauli strings that apply quantum gates on a restricted subset of qubits, allowing for simultaneous synthesis of the whole group and reducing both T-gate count and depth. Experimental results demonstrate that NCF achieves an average reduction of 57.4\%, 49.1\%, and 49.0\% in T-gate count, T-gate depth, and Clifford count, respectively, compared to the state-of-the-art method.

\end{abstract}


\section{Introduction}

Among various quantum algorithms, Hamiltonian simulation is a fundamental approach for modeling complex systems such as quantum chemistry~\cite{poulin2014trotter} and many-body physics~\cite{jiang2018quantum, georgescu2014quantum}.
Given a Hamiltonian $H$ of the target system and an evolution time $t$, the goal is to implement the operator $U = e^{-iHt}$. Since directly realizing $e^{-iHt}$ as a quantum circuit is challenging, the operator is typically decomposed into a sum of simpler Hamiltonians~\cite{lloyd1996universal} using Trotterization~\cite{trotter1959product}, where each simpler Hamiltonian corresponds to a Pauli string. Each Pauli string is then decomposed into a sequence of quantum gates, and repeatedly executing these gate sequences enables the simulation of the target system.

In general, implementing each Pauli string requires an RZ rotation gate along with a list of Clifford gates. On most NISQ quantum computers, the RZ gate is natively supported and can be executed virtually~\cite{ibm-virtual-rz} or with relatively low error rates, whereas two-qubit Clifford gates typically suffer from higher error rates~\cite{ibm_quantum, chen-2024-ionq-gate-times, Rigetti}. Consequently, current compilation frameworks for Hamiltonian simulation primarily focus on reducing the number of two-qubit Clifford gates, such as CNOT gates~\cite{debrugière2024faster, li2022paulihedral, liu2025quclear}.  
However, these compilation frameworks cannot be directly applied to fault-tolerant quantum computers, as the optimization objectives differ. In a fault-tolerant setting that leverages quantum error-correcting codes, each rotation gate, such as an RZ gate, is typically synthesized into a sequence of Clifford and T gates~\cite{ross2014optimal, paradis2024synthetiq, hao2025reducing, litinski2019game, fowler2012surface, kliuchnikov2022shorter}. Clifford gates can usually be implemented with relatively low overhead, requiring only a one or two error-correction cycle and a small number of ancilla qubits~\cite{fowler2018low, litinski2019game}. In contrast, T gates rely on magic state distillation, a resource-intensive process that is at least one magnitude more expensive in both time and qubit count per high-fidelity magic state than Clifford gates~\cite{bravyi2005universal, litinski2019game, fowler2018low, bravyi2012magic, beverland2021cost}. Therefore, \textbf{T-gate count} and \textbf{T-gate depth} emerge as the key optimization metrics for fault-tolerant quantum computers, rather than the Clifford count (e.g., CNOT gates).

In general, to implement a Hamiltonian simulation algorithm, the RZ gate in each Pauli string is synthesized using a dedicated RZ gate synthesizer such as Gridsyn~\cite{ross2014optimal}. In addition to RZ synthesizers, arbitrary single-qubit rotation (i.e., U3) synthesizers such as Trasyn~\cite{hao2025reducing} and multi-qubit unitary synthesizers such as Synthetiq~\cite{paradis2024synthetiq} can be used to synthesize non-Clifford gates. Since the number of T gates synthesized by an RZ synthesizer is similar to that of U3 and multi-qubit unitary synthesizers~\cite{hao2025reducing,kliuchnikov2022shorter}, one can reduce both T-gate count and T-gate depth by merging multiple RZ gates into a U3 gate or a two-qubit unitary and then synthesizing the merged unitary. However, to the best of our knowledge, no existing compilation framework has been proposed to optimize Hamiltonian simulation through such merging of RZ gates into larger unitaries.

In this paper, we present Non-Clifford Fusion (NCF), a compilation framework designed to reduce both T-gate count and T-gate depth by fusing non-Clifford rotations into compact single- or two-qubit unitaries. The framework overview of NCF is shown in Fig.~\ref{fig:framework}. NCF adopts a two-stage design: (i) it first partitions multiple Pauli strings into a group such that they can be simultaneously conjugated (i.e., transformed by a Clifford circuit) into a new set of Pauli strings that apply quantum gates only on one or two qubits, and (ii) it then generates the Clifford circuit along with the conjugated Pauli strings. By applying the Clifford circuit followed by the conjugated Pauli strings, NCF produces a quantum circuit that implements the operator $U = e^{-iHt}$ of the Hamiltonian simulation. Within each group, a single-qubit or two-qubit synthesizer is then employed to simultaneously synthesize the conjugated Pauli strings into Clifford and T gates. This group-wise synthesis strategy allows NCF to capture optimization opportunities that existing frameworks overlook, leading to significant reductions in both T-gate count and T-gate depth.


Our contributions are as follows:
\squishlist{}
    \item We explore the opportunity to reduce both T-gate count and T-gate depth in Hamiltonian simulation by fusing non-Clifford rotations into single-qubit or two-qubit unitaries, and synthesizing each fused unitary rather than synthesizing individual RZ gates. Based on this opportunity, we propose NCF, a compilation framework designed to reduce both T-gate count and T-gate depth for Hamiltonian simulation in the fault-tolerant era.
    \item NCF can serve as a generalized framework for evaluating different circuit synthesis tools across varying unitary sizes. We evaluated NCF with multiple synthesizers, revealing the advantages and limitations of existing synthesis methods.
    \item We evaluate NCF on 13 representative Hamiltonian simulation benchmarks. On average, single-qubit NCF with arbitrary U3 gate synthesizer Trasyn~\cite{hao2025reducing} achieves reductions of 57.4\% in T-gate count, 49.1\% in T-gate depth, and 49.0\% in Clifford count compared to the state-of-the-art framework Rustiq~\cite{debrugière2024faster} with Trasyn.
    
\squishend{}

\begin{figure}
    \centering
   \includegraphics[width=1\linewidth]{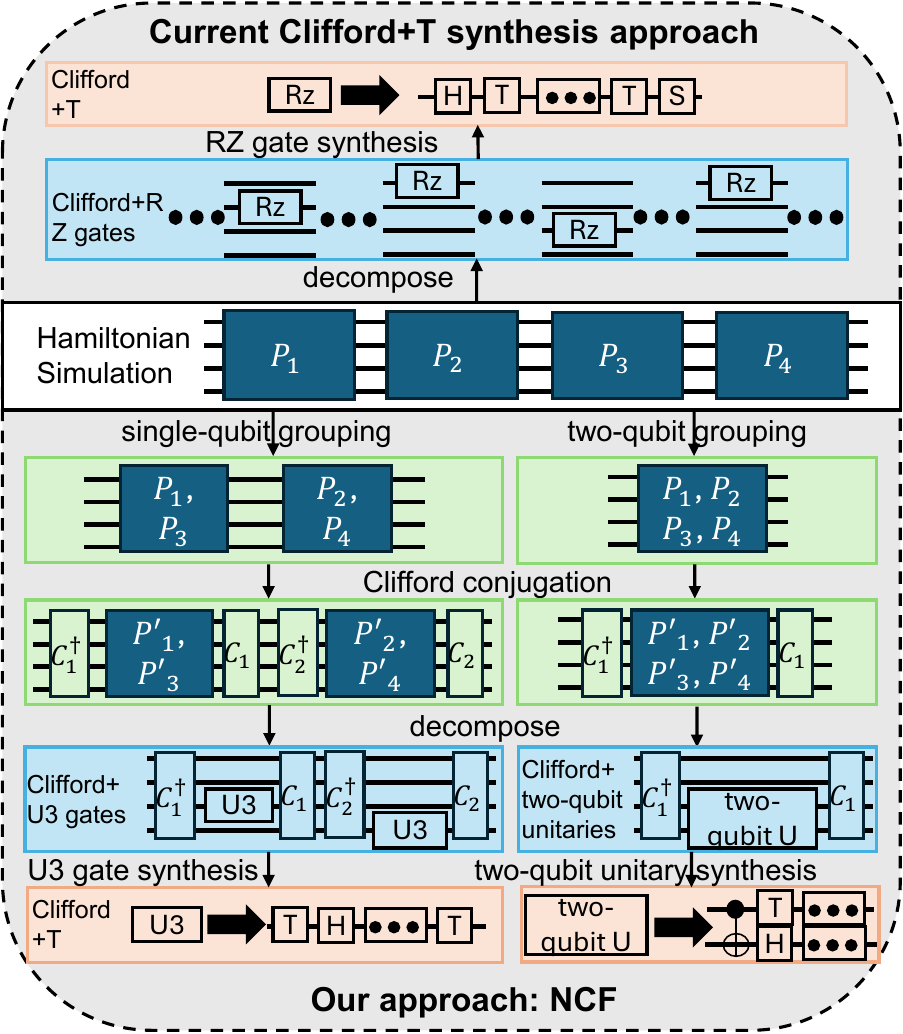}
     \caption{Overviews of NCF and the baseline.}
    \label{fig:framework}
 \end{figure}

\section{Background}
\subsection{Hamiltonian Simulation}
Hamiltonian simulation is a promising approach for modeling complex problems in physics~\cite{jiang2018quantum} and chemistry~\cite{poulin2014trotter}. 
Its goal is to capture the time evolution of a quantum system governed by a Hamiltonian $H$, represented by the unitary time-evolution operator $e^{-iHt}$ on a quantum circuit, where $t \in \mathbb{R}$ denotes the evolution time. By using the Trotter decomposition~\cite{trotter1959product}, this operator can be approximated as $e^{-i H t} \approx \left( \prod_{j=1}^m e^{-i w_j P_j \Delta t} \right)^{\frac{t}{\Delta t}}$, where the Hamiltonian $H$ is decomposed into a sum of simpler terms $H = \sum_{j=1}^m w_j P_j$ and $\Delta t$ is the timestep to control the precision of the approximation~\cite{lloyd1996universal}. Here, $w_j \in \mathbb{R}$ is the coefficient of the $j_{th}$ term and $P_j$ is a Pauli String. Fig.~\ref{fig:hamil} shows an example of a Hamiltonian simulation and its decomposition.

\begin{figure}
    \centering
   \includegraphics[width=1\linewidth]{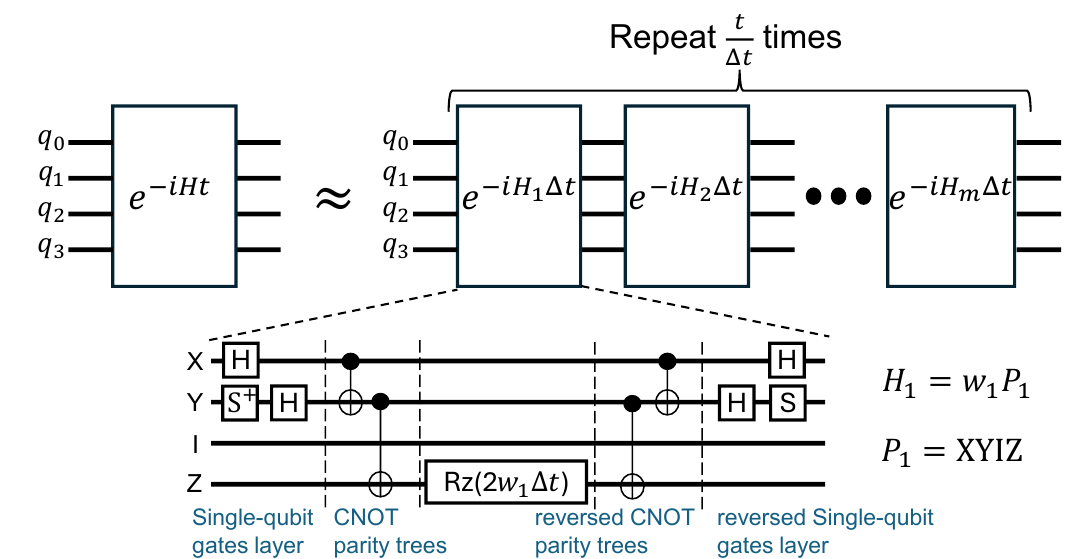}
     \caption{An example of Hamiltonian simulation and Pauli String circuit.}
    \label{fig:hamil}
 \end{figure}

\subsection{Pauli String}
\label{sec:pauli}
After decomposing the time-evolution operator into a series of exponential terms, each term $e^{-i w_j P_j \Delta t}$ can be implemented as a sequence of quantum gates~\cite{li2022paulihedral, liu2025quclear, mukhopadhyay2023synthesizing}. In an $n$-qubit system, each Pauli String $P_j=\sigma_{0} \otimes \sigma_{1} \otimes...\otimes \sigma_{n-1}$ is a tensor product of $n$ single-qubit operators, where each $\sigma_i \in \{X,Y,Z,I\}$. Here, $I$ denotes the identity operator, while $X,Y,$ and $Z$ are the Pauli operators. In the remainder of this paper, we omit the tensor product symbol $\otimes$ when writing Pauli strings. We define ``acting trivially'' on a qubit if the corresponding operator in a Pauli string is $I$, meaning no quantum gates need to be applied to that qubit. Conversely, a qubit ``acts non-trivially'' if the operator is non-$I$ (i.e., $X$, $Y$, and $Z$), indicating that quantum gates must be applied.

An example of a building-block circuit for implementing $e^{-i w_1 P_1 \Delta t}$ is shown in Fig.~\ref{fig:hamil}, where the Pauli string is $P_1=XYIZ$. It begins with a layer of single-qubit gates that map all $X$ and $Y$ operators to the $Z$ operators (i.e., applying an $H$ gate for $X \rightarrow Z$ and $S^{\dagger}H$ for $Y \rightarrow Z$). A chain of CNOT gates is then applied to form a parity aggregation tree that entangles all qubits on which $P_j$ acts non-trivially. Next, we perform a single-qubit rotation $RZ(2w_{j}\Delta t)$ on the target qubit of the CNOT tree to implement the phase evolution. Finally, the CNOT tree and basis-change gates are reversed to restore the original basis. By applying this procedure to all terms in the Trotter step, we obtain an approximation of the Hamiltonian simulation.

\subsection{Tableau Representation}

For a Hamiltonian simulation involving $m$ Pauli strings and $n$ qubits, the collection of operators can be compactly represented as an $m \times (2n+1)$ binary matrix, known as the tableau representation~\cite{van2021simple, debrugière2024faster, aaronson2004improved}. Each row of the tableau corresponds to a Pauli string, while the first $2n$ columns encode its action on the $n$ qubits: the first $n$ columns encode the $X$ operators while the next $n$ columns encode the $Z$ operators, which are referred to as the X matrix and Z matrix, respectively. The last column $r$ is used to track the sign of the Pauli string (i.e., + or -). An $Y$ operator is represented by having both the $X$ and $Z$ entries set to 1 for the same qubit. An example of the tableau of two Pauli strings are shown in Fig.~\ref{fig:pauli}(a). Thus, a Hamiltonian simulation circuit can be represented by its tableau representation together with the associated rotation angle for each Pauli string.

\subsection{Commutation and Anticommutation}
\label{sec:commute}
Pauli strings obey commutation relations based on the commutation rules of single-qubit Pauli operators~\cite{gokhale2019minimizing, van2021simple}. Each Pauli operator commutes with itself and with the identity operator $I$ (e.g., $X$ commutes with both $X$ and $I$), while it anticommutes with the other two non-identity Pauli operators (e.g., $X$ anticommutes with $Z$ and $Y$). Extending this to multi-qubit Pauli strings, two Pauli strings either commute or anticommute depending on the number of qubits where the corresponding Pauli operators anticommute. If this number is even, the strings commute; if odd, they anticommute.

For example, consider three Pauli strings on two qubits: $P_1=XZ$, $P_2=YI$, and $P_3=YX$. $P_1$ anticommutes with $P_2$ because their Pauli operator on the first qubit anticommute (i.e, $X$ and $Y$) and that on the second qubit commute (i.e., $Z$ and $I$). Since there is only one qubit position where they anticommute, the two Pauli strings anticommute. In contrast, $P_1$ commutes with $P_3$ becasue the Pauli operators on both qubits anticommute. Because the number of anticommutes is even, the two Pauli strings commute. In Hamiltonian simulation via Trotterization, the order of Pauli strings can still be exchanged regardless of commutation or anticommutation~\cite{li2022paulihedral, debrugière2024faster}. The effect of reordering impacts only the approximation error in Trotterization, which can be mitigated by decreasing the time step $\Delta t$~\cite{trotter1959product, georgescu2014quantum}.

\begin{figure}
    \centering
   \includegraphics[width=1\linewidth]{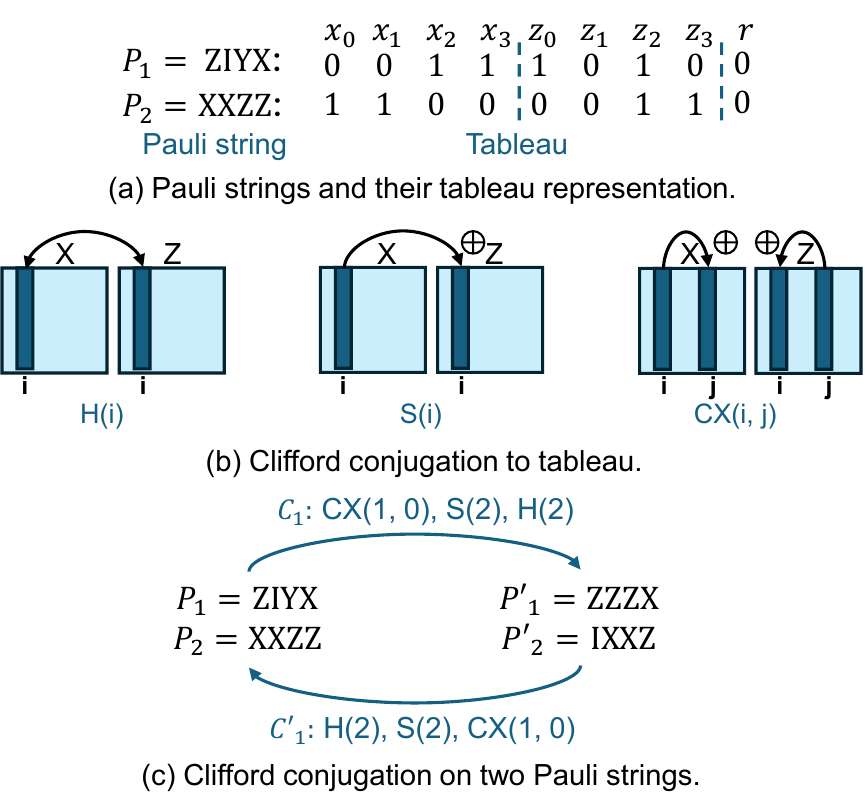}
     \caption{Examples of Tableau representation and Clifford conjugation.}
    \label{fig:pauli}
 \end{figure}

\begin{table}[h]
\centering
\caption{Pauli operator conjugation rules under Clifford gates. }
\begin{tabular}{c|c|c}
\hline
\textbf{Gate} & \textbf{Input Pauli} & \textbf{Output Pauli} \\
\hline
\multirow{3}{*}{Hadamard ($H$)} 
& $X$ & $Z$ \\
& $Y$ & $-Y$ \\
& $Z$ & $X$ \\
\hline
\multirow{3}{*}{Phase ($S$)} 
& $X$ & $Y$ \\
& $Y$ & $-X$ \\
& $Z$ & $Z$ \\
\hline
\multirow{6}{*}{CNOT ($\mathrm{CNOT}_{c \rightarrow t}$)} 
& $X \otimes I$ & $X \otimes X$ \\
& $Y \otimes I$ & $Y \otimes X$ \\
& $Z \otimes I$ & $Z \otimes I$ \\
& $I \otimes X$ & $I \otimes X$ \\
& $I \otimes Y$ & $Z \otimes Y$ \\
& $I \otimes Z$ & $Z \otimes Z$ \\
\hline
\end{tabular}
\label{tab:clifford_conjugation}
\end{table}

\subsection{Generator and Generated Pauli Strings}
Another important aspect of Pauli strings is that a set of Pauli strings can be partitioned into generators and generated elements by using Gaussian elimination~\cite{gokhale2019minimizing, aaronson2004improved}. The generators form a subset of Pauli strings such that all other Pauli strings in the set can be obtained by multiplying subsets of these generators. For example, if $\{P_1, P_2, \ldots, P_k\}$ are generators, then any other Pauli string $Q$ in the set can be written as 
    $Q = \prod_{i \in S} P_i$,
for some subset $S \subseteq \{1,2,\ldots,k\}$. 

\subsection{Clifford Circuit Conjugation}
\label{sec:conju}
A Pauli string $P$ can be transformed to another Pauli string $P'$ under conjugation by a Clifford circuit $C$ (i.e., $C P C^\dagger = P'$)~\cite{van2021simple, debrugière2024faster, gokhale2019minimizing}. A Clifford circuit is composed solely of gates from the Clifford group (i.e., H, S, and CNOT gates). We summarize the conjugation rule for each Clifford gate in Table~\ref{tab:clifford_conjugation}. Additionally, since Pauli strings can be represented by a tableau, the conjugation of a Clifford gate can be represented as an update to the corresponding tableau of Pauli strings, as shown in Fig.~\ref{fig:pauli}(b). Specifically, they are:

\squishlist{}
    \item \textbf{H gate}: Conjugation by an $H$ gate swaps $X$ and $Z$ on the acted qubit while flipping the sign of $Y$. 
    In the tableau, this corresponds to swapping the $x_i$ and $z_i$ columns for qubit $i$. 

    \item \textbf{S gate}: Conjugation by an $S$ gate maps $X \mapsto Y$, $Y \mapsto -X$, and leaves $Z$ unchanged. 
    In the tableau, this corresponds to updating the $z_i$ column as $z_i \leftarrow z_i \oplus x_i$ (with $\oplus$ denoting XOR). 

    \item \textbf{CNOT gate}: Conjugation by a $\mathrm{CNOT}_{c \to t}$ propagates $X$ on the control qubit to the target and $Z$ on the target qubit to the control, while leaving $X$ on the target and $Z$ on the control unchanged. 
    In the tableau, this corresponds to the update rules:
    \[
    x_t \leftarrow x_t \oplus x_c, \qquad z_c \leftarrow z_c \oplus z_t.
    \]
    Since $Y$ is the product of $X$ and $Z$, its transformation follows directly from these update rules.
\squishend{}

There are four important rules in Clifford conjugation. 

\textbf{Rule-I:} multiple Pauli strings can be conjugated simultaneously. That is because unitary conjugation preserves operator multiplication; that is, for a Clifford circuit $C$, and Pauli Strings $P_1$ and $P_2$, 

\begin{equation}
C \,(P_1 P_2)\, C^\dagger \;=\; (C P_1 C^\dagger) \,(C P_2 C^\dagger)
\label{eq:clifford_conjugation_product}
\end{equation}
so the same conjugation layer applied once transforms every Pauli string in the product individually. When representing a list of Pauli strings using a tableau, one or more columns of the tableau are updated for each gate in the Clifford circuit. 

\textbf{Rule-II:} Pauli strings can always be transformed back to their original form by applying the inverse Clifford circuit. Since Clifford conjugation is a unitary operation, it is reversible: for a Clifford circuit $C$, a Pauli string $P$, and the conjugated Pauli string $P'$, we have  
\[
C P C^\dagger = P', \qquad C^\dagger P' C = P,
\]  
An example of Clifford conjugation on two Pauli Strings and then transform back to the original form is shown in Fig.~\ref{fig:pauli}(c). Therefore, applying the conjugated Pauli string with the Clifford circuit realizes the same operation as the original Pauli string.

\textbf{Rule-III:} The commutation relations between Pauli strings are preserved under Clifford conjugation. Specifically, if two Pauli strings $P_1$ and $P_2$ commutes, then their conjugated Pauli strings under the same Clifford circuit $C$, denoted as $P'_1=CP_1 C^\dagger$ and $P'_2=CP_2 C^\dagger$, also commutes. 

\textbf{Rule-IV:} When applying the same Clifford circuit to conjugate a list of different Pauli strings, the resulting conjugated Pauli strings will also be distinct (i.e., no two Pauli strings will be mapped to the same conjugated Pauli string).

\section{Motivation and Opportunity}
\label{sec:motivation}
In fault-tolerant quantum computers, non-Clifford Unitaries—such as RZ gates and arbitrary single-qubit rotation gates (i.e., U3 gate) are typically decomposed into sequences of Clifford and T gates~\cite{ross2014optimal, hao2025reducing, paradis2024synthetiq, litinski2019game, kliuchnikov2022shorter, fowler2012surface}. Clifford gates can often be implemented transversally, requiring only a single or two code cycles (i.e., the time to perform error correction on all physical qubits) and minimal ancilla overhead~\cite{fowler2012surface, litinski2019game}. In contrast, T gates cannot be implemented transversally in most codes and instead rely on magic state distillation~\cite{fowler2012surface, bravyi2005universal}, a resource-intensive process that is at least one magnitude more expensive in both time and physical qubit count than a Clifford gate~\cite{beverland2021cost, litinski2019game}, making \textbf{T-gate count} and \textbf{T-gate depth} dominant metrics in fault-tolerant circuit optimization. Due to the large number of physical qubits required, the availability of Magic State Factories (MSFs) is limited, making the T-gate count a primary bottleneck~\cite{litinski2019game, fowler2012surface, bravyi2005universal, beverland2021cost}. Although T gates can, in principle, be executed in parallel, the restricted number of MSFs and the cycles required to distill each T state significantly constrain this parallelism. As fault-tolerant architectures advance, with faster T-gate production and larger numbers of MSFs, the degree of parallelism will increase, and the performance bottleneck could shift from T-gate count to T-gate depth~\cite{litinski2019game}.

The decomposition of non-Clifford unitaries requires a synthesizer that produces a sequence of Clifford and T gates approximating the target unitary to within a specified precision $\epsilon$. For single-qubit rotations, the state-of-the-art synthesizer for RZ gates is Gridsynth~\cite{ross2014optimal}, which produces a series of quantum gates with T-gate count $3\log_2({\frac{1}{\epsilon}})$ (i.e., the lower the $\epsilon$, the more T gates required). For arbitrary single-qubit unitaries such as the U3 gate, Trasyn~\cite{hao2025reducing} achieves a comparable T-gate cost to Gridsynth. Beyond single-qubit unitaries, methods have also been developed for synthesizing multi-qubit unitaries. The state-of-the-art multi-qubit Synthesizer is Synthetiq~\cite{paradis2024synthetiq}, which is capable of synthesizing unitaries with up to four qubits. As shown in \cite{kliuchnikov2022shorter}, the two-qubit synthesizer can achieve a T-gate count of $11.5 \log_2\!\left(\frac{1}{\epsilon}\right)$, which is comparable to the T-gate count of the single-qubit synthesizer. Note that the T-gate count produced by these synthesizers depends solely on the error precision $\epsilon$, and is independent of the specific unitary being synthesized~\cite{ross2014optimal,kliuchnikov2022shorter,hao2025reducing}.

Using the U3 gate synthesizer, one way to reduce both the T-gate count and T-gate depth is to merge a sequence of consecutive single-qubit gates into a single U3 gate and synthesize each U3 gate, rather than synthesizing each RZ gate independently. This reduction arises from two factors: (i) Since the T-gate count scaling of the RZ and U3 gate synthesizers is similar, combining $n$ RZ gates into a single U3 gate reduces the number of T gates to approximately $\frac{1}{n}$ of that required when synthesizing each RZ gate individually. (ii) Synthesizing each RZ gate with an error threshold $\epsilon$ introduces a cumulative error of $n\epsilon$. By merging $n$ RZ gates, the error threshold for synthesis can be relaxed from $\epsilon$ to $n\epsilon$ (i.e., a higher error threshold), further reducing both the T-gate count and T-gate depth for the same total error. Extending this idea, multi-qubit synthesizers enable the merging of more RZ gates across multiple qubits into a multi-qubit unitary. 

However, to the best of our knowledge, no existing work specifically targets the merging of single-qubit rotation gates to reduce both the T-gate count and T-gate depth. In the current Hamiltonian simulation implementation~\cite{liu2025quclear, li2022paulihedral}, the number of RZ gates is directly proportional to the number of Pauli strings, as illustrated in Fig.~\ref{fig:hamil}. This direct correspondence results in both high T-gate count and substantial T-gate depth.

The only method that enables the merging of RZ gates is Rustiq~\cite{debrugière2024faster}, a state-of-the-art Hamiltonian simulation compiler designed primarily to group multiple Pauli strings and use Clifford circuit conjugation in order to reduce the number of CNOT gates. This grouping strategy accidentally allows some RZ gates to be merged into U3 gates, thereby lowering the T-gate count. It also enables parallel execution of RZ rotations to reduce the T-gate depth. However, the reduction of T-gate resources is merely a byproduct of the CNOT optimization, rather than a targeted objective of the approach. 

Merging RZ gates in Hamiltonian simulation is not a trivial task. As discussed in Section~\ref{sec:pauli}, each RZ gate is sandwiched between two CNOT trees, which prevents merging unless two or more Pauli strings act non-trivially on the same qubits. Furthermore, determining which Pauli strings to merge in order to minimize both the T-gate count and T-gate depth is challenging, due to the large number of possible merging combinations.
 This motivates the question: \textit{How can we merge the RZ gates for Hamiltonian simulation to reduce both the T-gate count and T-gate depth?}

\section{Our Method}
\label{sec:method}
In this paper, we propose Non-Clifford Fusion (NCF), a compilation framework designed to reduce both the T-gate count and T-gate depth in Hamiltonian simulation. NCF adopts a two-stage design. In the first stage, the input set of Pauli strings is partitioned into groups such that the Pauli strings within each group can be conjugated into a new set acting non-trivially on only one or two qubits. In the second stage, we construct Clifford circuits that conjugate the Pauli strings in each group, yielding the conjugated Pauli strings. Applying these conjugated Pauli strings together with the Clifford circuits reproduces the same functionality as the original Pauli strings. Within each group, the RZ gates corresponding to the conjugated Pauli strings can then be merged into a single unitary, which is synthesized using either a U3 gate synthesizer or a two-qubit unitary synthesizer.


To control the complexity of NCF, we introduce a window strategy that restricts the number of Pauli strings considered for inclusion in a group, which limits the search space while still capturing effective merging opportunities.

\subsection{NCF Grouping}
\label{sec:grouping}
In the first stage of NCF, Pauli strings are partitioned into groups such that the RZ gates within each group can be merged into either a single-qubit or a two-qubit unitary. Although merging into higher-qubit unitaries is possible, the scalability limitations of the synthesizers make the efficient synthesis of unitaries involving more than two qubits impractical~\cite{hao2025reducing, paradis2024synthetiq}. Moreover, the T-gate count scaling of higher-qubit synthesizers has not yet been established~\cite{kliuchnikov2022shorter}, leaving the potential benefits of using such synthesizers uncertain. Therefore, NCF focuses on single-qubit and two-qubit unitaries.

\subsubsection{Single-qubit Grouping}
As discussed in \cite{van2021simple}, any pair of anticommuting Pauli strings can be simultaneously conjugated into two new anticommuting Pauli strings that acts non-trivially on the same qubit, with distinct Pauli operators on that qubit (e.g., $X$ and $Z$). For example, the two Pauli strings shown in Fig.~\ref{fig:pauli}(c) can be conjugated into two Pauli strings $P'_1 = XIII$ and $P'_2=ZIII$. This follows Rule-III in Section~\ref{sec:conju}, the commutation relationship between Pauli strings is preserved under Clifford conjugation. According to the circuit construction rules for Pauli strings described in Section~\ref{sec:pauli}, no CNOT parity tree is required in this case since only one qubit is acted on non-trivially. Consequently, the single-qubit gates for both Pauli strings are applied to the same qubit, allowing them to be merged into a single unitary.


Moreover, if a Pauli string can be generated by two selected Pauli strings, all three Pauli strings can be conjugated to act non-trivially on a single qubit, enabling their merging into a single unitary operation. This is possible because: 
\begin{lemma}
Let $C$ be a Clifford circuit and let $P_1,\dots,P_m$ be Pauli strings on $n$ qubits. 
If the conjugated Pauli strings $P'_j := C P_j C^\dagger$ act non-trivially only on a fixed subset of qubits $S \subseteq \{1,\dots,n\}$ for $j \in \{1, \cdots, m \}$, 
then for any product $Q=\prod_{j=1}^m P_j^{a_j}$ with $a_j\in\{0,1\}$, the conjugated term $Q' := C Q C^\dagger$ also acts non-trivially only on $S$.
\end{lemma}

\begin{proof}
By Rule-I in Section~\ref{sec:conju},
\[
Q' = C \Big(\prod_{j=1}^m P_j^{a_j}\Big) C^\dagger 
   = \prod_{j=1}^m \big(C P_j C^\dagger\big)^{a_j}
   = \prod_{j=1}^m (P'_j)^{a_j}.
\]
Each $P'_j$ is trivial outside $S$, so their product is also trivial outside $S$. So $Q'$ only act non-trivially on $S$.
\end{proof}

Furthermore, at most three Pauli strings can be simultaneously conjugated to act non-trivially on the same single qubit. This follows from Rule-IV in Section~\ref{sec:conju}, which ensures that different Pauli strings are conjugated into distinct Pauli operators. Since only three non-identity Pauli operators ($X$, $Y$, and $Z$) can act on a single qubit, no more than three Pauli strings can be included in a single group.

In this stage, we search for pairs of anti-commuting Pauli strings together with the Pauli strings they generate, and partition them into a group. Our method proceeds as follows: for the set of Pauli strings in a Hamiltonian simulation, we first construct a commuting graph and an anti-commuting graph, where vertices represent Pauli strings and edges capture whether two strings commute or anticommute, respectively.

For Pauli strings that are not yet partitioned, we first perform Gaussian elimination to identify a set of generators and the corresponding generated Pauli strings. Using the anti-commuting graph, we then extract all pairs of anti-commuting generator Pauli strings. For each candidate pair, we check whether their generated Pauli string also appears in the Hamiltonian simulation and remains unpartitioned. If so, we group the first pair with their generated Pauli string; otherwise, we select the first pair of anti-commuting Pauli strings.
This process continues iteratively until either all Pauli strings are partitioned or the remaining ones cannot be grouped further (i.e., they are mutually commuting). 

We then partition the remaining ungrouped, mutually commuting Pauli strings into groups. Although this grouping does not allow their $R_Z$ gates to be merged, these Pauli strings can be conjugated to act non-trivially on distinct single qubits, enabling their $R_Z$ gates to be executed in parallel and thereby reducing the T-gate depth~\cite{van2021simple, gokhale2019minimizing}. Since each commuting Pauli string is conjugated to act on a distinct qubit, the maximum number of Pauli strings in such a group is equal to the logical qubit count $q$.

To distinguish the two types of groups, we refer to the groups whose Pauli strings can be conjugated to act on a single qubit as ``anticommuting groups,'' and the groups with mutually commuting Pauli strings as ``commuting groups.'' After all Pauli strings are partitioned into groups, we further reduce the T-gate depth by reordering the groups to allow simultaneous execution of compatible groups. Starting from the first group, we search through the remaining groups and check whether any subsequent group can be executed concurrently, and if such a combination exists, we place the groups together.


An example of the grouping process is shown in Fig.~\ref{fig:NCF}. Fig.~\ref{fig:NCF}(a) illustrates six Pauli strings, while Fig.~\ref{fig:NCF}(b) shows the default quantum circuit to implement them. For the six Pauli strings in Fig.~\ref{fig:NCF}(a), we begin by applying Gaussian elimination to distinguish the generators from the generated Pauli strings. Here, $P_1, P_2, P_3, P_4,$ and $P_6$ are the generators, while $P_5$ is generated by $P_1$ and $P_3$. We first partition $P_1, P_3,$ and $P_5$ into a group, as these can be conjugated to act on a single qubit. In the next iteration, we group $P_4$ and $P_6$ since they are the next anti-commuting pair. Finally, for the remaining $P_2$, we set it as a single group. The final groups are shown in Fig.~\ref{fig:NCF}(c). 

Using the Clifford circuit generation approach that will be introduced in Section~\ref{sec:clifford}, we conjugate the six Pauli strings into six new Pauli strings with two Clifford circuits $C_1$ and $C_2$, as shown in Fig.~\ref{fig:NCF}(d). 
After forming the groups, we reorder them so that $P'_2$ can be executed simultaneously with group 1(i.e., $P'_1$, $P'_3$, and $P'_5$). The resulting quantum circuit, shown in Fig.~\ref{fig:NCF}(e), performs the same function as the original circuit in Fig.~\ref{fig:NCF}(b). This process reduces the number of unitaries from six to three, significantly decreasing both the T-gate counts and T-gate depth.

 \begin{figure*}
    \centering
   \includegraphics[width=1\linewidth]{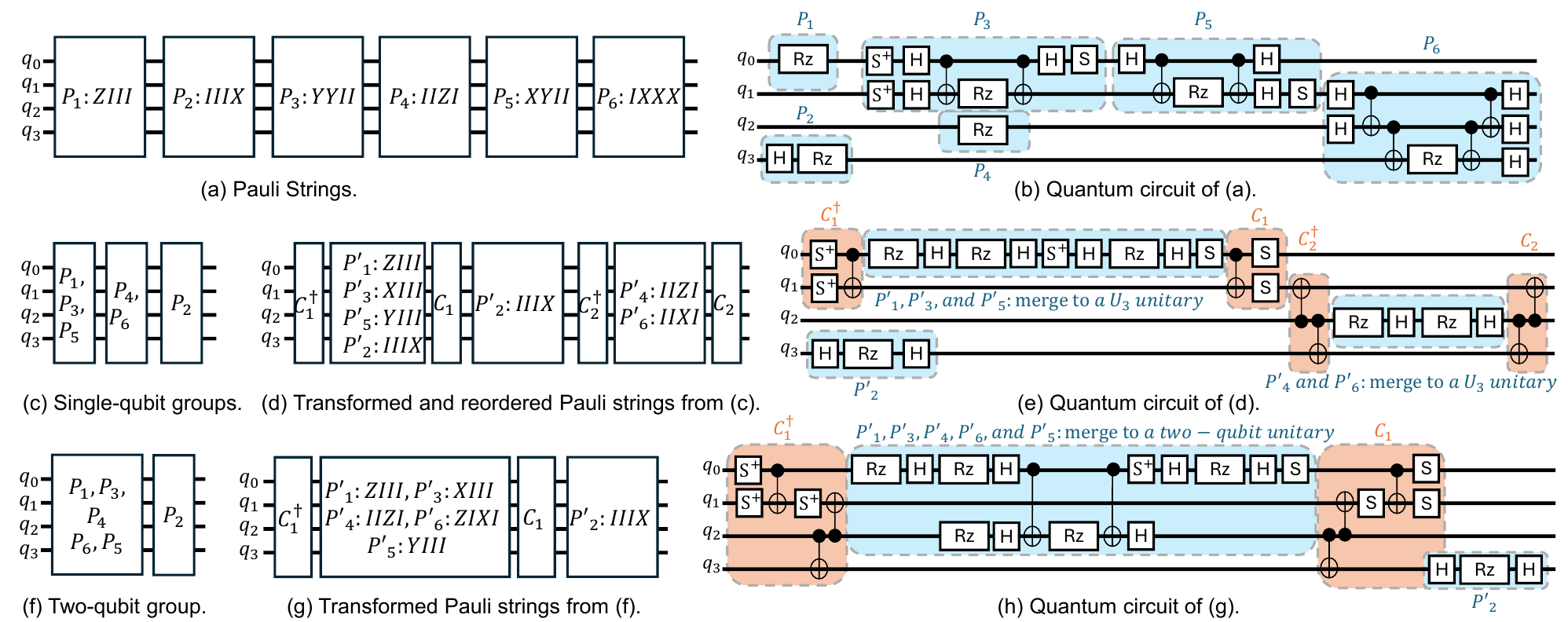}
     \caption{Examples of NCF applied to six Pauli strings for both single and two-qubit unitaries.}
    \label{fig:NCF}
 \end{figure*}

\begin{figure}
    \centering
   \includegraphics[width=1\linewidth]{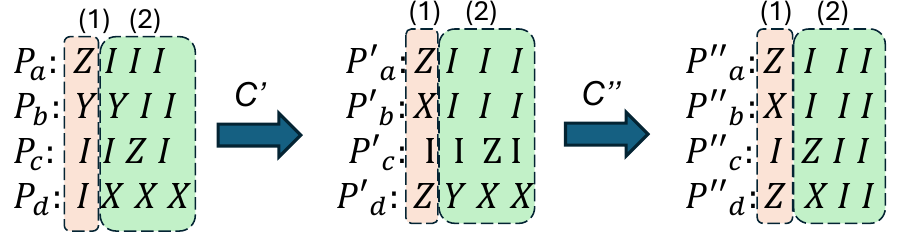}
     \caption{An example of conjugating four Pauli strings.}
    \label{fig:two}
 \end{figure}

\subsubsection{Two-qubit Grouping}
Similar to single-qubit grouping, our two-qubit grouping method searches for two pairs of generator Pauli strings and their generated Pauli strings that can be conjugated to act non-trivially on two qubits. The method begins by identifying a pair of anticommuting Pauli strings. As discussed in the single-qubit grouping, this pair can be conjugated by a Clifford circuit $C'$ to act non-trivially on a single qubit. We then search for a second pair of Pauli strings which, after conjugation by $C'$ followed by another Clifford circuit $C''$, act non-trivially on two qubits, including the qubit already involved in the first pair. In other words, the four Pauli strings along with their generated Pauli strings can be simultaneously conjugated by the combined Clifford circuit $C=C'+C''$ to act non-trivially on two qubits.

To distinguish the two pairs of Pauli strings,  we label the Pauli strings in the first anticommuting pair as $P_a$ and $P_b$, and those in the second pair as $P_c$ and $P_d$. 
For example, Fig.~\ref{fig:two} illustrates four four-qubit Pauli strings, $P_a$ to $P_d$, which correspond to $P_1$, $P_3$, $P_4$, and $P_6$ in Fig.~\ref{fig:NCF}, respectively. The four Pauli strings can be conjugated into $P'_a$ to $P'_d$ by a Clifford circuit $C'$. Here, we separate each Pauli string into two parts to better illustrate our method, where part (1) represents the qubit on which $P'_a$ and $P'_b$ act non-trivially. For instance, $P'_a(1) = Z$ and $P'_a(2) = III$ in Fig.~\ref{fig:two}. If $P'_c(2)$ and $P'_d(2)$ are also anti-commuting, they can be conjugated into $P''_c(2)$ and $P''_d(2)$ by a second Clifford circuit $C''$, where $P''_c(2)$ and $P''_d(2)$ act non-trivially on only one qubit, as illustrated in Fig.~\ref{fig:two}. After conjugating by both $C = C' + C''$, the four Pauli strings $P''_a$ to $P''_d$ act non-trivially on only two qubits. 

As introduced above, to identify the second pair of generator Pauli strings 
$P_c$ and $P_d$, one naive approach is to first construct a Clifford circuit 
$C'$ for the chosen pair $P_a$ and $P_b$, apply it to all remaining generator 
Pauli strings, and then analyze the commutation relations among the transformed 
Pauli strings to locate the appropriate $P_c$ and $P_d$. This procedure, however, 
incurs high computational complexity. However, we found that the second pair of Pauli strings can be identified solely by analyzing the commutation relations among $P_a$ through $P_d$, without explicitly generating and applying $C'$, and that this analysis needs to be performed only once, thereby reducing the complexity.
Based on the commutation checks, a truth table can be used to efficiently locate $P_a$ to $P_d$. The reason is:

Since each Pauli string is separated into two parts, (1) and (2), two Pauli strings anticommute if exactly one of the two parts anticommutes. 
Using this property, to determine whether $P'_c(2)$ and $P'_d(2)$ anti-commute, we exam the commutation relation between (i) $P'_c(1)$ and $P'_d(1)$, and (ii) $P'_c$ and $P'_d$. Specifically, $P'_c(2)$ and $P'_d(2)$ anti-commute if either $P'_c$ and $P'_d$ anti-commute while $P'_c(1)$ and $P'_d(1)$ commute, or $P'_c$ and $P'_d$ commute while $P'_c(1)$ and $P'_d(1)$ anti-commute. This relationship is summarized in the truth table shown in Table~\ref{tab:pauli_comm_full}. Based on Rule-III in Section~\ref{sec:conju}, the commutation relation between Pauli strings is preserved when conjugated by the same Clifford circuit. Therefore, the commutation relation between $P'_c$ and $P'_d$ can be determined directly by analyzing between $P_c$ and $P_d$.

\begin{table}[h]
\centering
\caption{Commutation conditions for $P'_c(2)$ and $P'_d(2)$}
\begin{tabular}{c|c|c}
\hline
$P'_c$($P_c$) vs $P'_d$($P_d$) & $P'_c(1)$ vs $P'_d(1)$ & $P'_c(2)$ vs $P'_d(2)$ \\
\hline
\textbf{Anti-commute} & \textbf{Commute}       & \textbf{Anti-commute} \\
 \textbf{Commute}     & \textbf{Anti-commute}  & \textbf{Anti-commute} \\
Anti-commute & Anti-commute  & Commute \\
Commute     & Commute       & Commute \\
\hline
\end{tabular}
\label{tab:pauli_comm_full}
\end{table}

Next, we determine the commutation relationship between $P'_c(1)$ and $P'_d(1)$. Since $P'_a(2)$ and $P'_b(2)$ are $I$ operators, they commute with both $P'_c(2)$ and $P'_d(2)$. Consequently, the commutation relations between $P'_a(1)$ and $P'_c(1)$, $P'_a(1)$ and $P'_d(1)$, as well as between $P'_b(1)$ and $P'_c(1)$, $P'_b(1)$ and $P'_d(1)$, can be fully determined by the original Pauli strings $P_a$ to $P_d$. We observed that this information is sufficient to establish the commutation relation between $P'_c(1)$ and $P'_d(1)$. Specifically, there are three possible conditions:

\squishlist{}
    \item \textbf{Condition-I: If $P_c$ or $P_d$ commutes with both $P_a$ and $P_b$, then $P'_c(1)$ and $P'_d(1)$ must commute.} As discussed above, if $P_c$ (or $P_d$) commutes with both $P_a$ and $P_b$, then after conjugation, $P'_c(1)$ (or $P'_d(1)$) also commutes with both $P'_a(1)$ and $P'_b(1)$. Since $P'_a(1)$ and $P'_b(1)$ are distinct non-$I$ operators acting on the same qubit, $P'_c(1)$ (or $P'_d(1)$) must be the $I$ operator. Because the identity operator commutes with every operator, it follows that $P'_c(1)$ and $P'_d(1)$ commute. For example, in Fig.~\ref{fig:two}, since $P_c$ commutes with both $P_a$ and $P_b$, we have $P'_c(1) = I$, which therefore commutes with $P'_d(1)$.

    \item \textbf{Condition-II: $P'_c(1)$ and $P'_d(1)$ commutes if $P_c$ and $P_d$ has the same commutation relationship between $P_a$ and $P_b$.} For $P_c$ and $P_d$, if they share the same commutation relations with $P_a$ and $P_b$, then under the same Clifford conjugation, $P'_c(1)$ and $P'_d(1)$ must correspond to the same Pauli operator, and thus commute. This is because $P'_a(1)$ and $P'_b(1)$ are distinct non-$I$ operators, there is only one non-$I$ Pauli operator that satisfies the same commutation relations with two distinct non-$I$ operators~\cite{van2021simple, gokhale2019minimizing}. 

    \item \textbf{Condition-III: $P'_c(1)$ and $P'_d(1)$ anticommute if the above two conditions are not met.} $P'_c(1)$ and $P'_d(1)$ anticommute when they are distinct non-$I$ operators. Condition-I covers the case where either $P'_c(1)$ or $P'_d(1)$ is an $I$ operator, and Condition-II covers the case where $P'_c(1)$ and $P'_d(1)$ are the same non-$I$ operator. Therefore, if neither of these conditions applies, $P'_c(1)$ and $P'_d(1)$ must anticommute.
\squishend{}

Using the three conditions above along with the commutation relationship between $P_c$ and $P_d$, we can determine whether $P'_c(2)$ and $P'_d(2)$ anticommute. Consequently, the Pauli strings $P_a$ to $P_d$ can be conjugated into $P''_a$ to $P''_d$, which act non-trivially on only two qubits. This enables us to use a truth table to select the appropriate $P_a$ to $P_d$. The corresponding truth table is shown in Table~\ref{tab:truth_table_PcPd}.

\begin{table}[h]
\centering
\caption{Truth table for selecting $P_a$ to $P_d$ to enable $P''_a$ to $P''_d$ act non-trivially on two qubits. Only the allowed combination is shown. 1 indicates anticommutes}
\begin{tabular}{c|c|c|c|c}
\hline
$P_a$ vs $P_c$ & $P_b$ vs $P_c$ & $P_a$ vs $P_d$ & $P_b$ vs $P_d$ & $P_c$ vs $P_d$\\
\hline
0 & 0 & 0 & 0 & 1 \\
0 & 0 & 0 & 1 & 1 \\
0 & 0 & 1 & 0 & 1 \\
0 & 0 & 1 & 1 & 1 \\
0 & 1 & 0 & 0 & 1 \\
0 & 1 & 0 & 1 & 1 \\
0 & 1 & 1 & 0 & 0 \\
0 & 1 & 1 & 1 & 0 \\
1 & 0 & 0 & 0 & 1 \\
1 & 0 & 0 & 1 & 0 \\
1 & 0 & 1 & 0 & 1 \\
1 & 0 & 1 & 1 & 0 \\
1 & 1 & 0 & 0 & 1 \\
1 & 1 & 0 & 1 & 0 \\
1 & 1 & 1 & 0 & 0 \\
1 & 1 & 1 & 1 & 1 \\

\hline
\end{tabular}
\label{tab:truth_table_PcPd}
\end{table}

Using the truth table, we now present our grouping strategy for two-qubit grouping. Similar to the single-qubit grouping method, we first construct commuting and anti-commuting graphs for the list of Pauli strings. At each iteration, for Pauli strings that have not yet been partitioned, we perform Gaussian elimination to identify a set of generator Pauli strings and the corresponding Pauli strings they generate. 
Our goal is to partition as many Pauli strings as possible into a group, which can then be merged into a two-qubit unitary. However, there are $ \binom{m}{4} = \frac{m!}{4!(m-4)!} $ possible combinations for $m$ Pauli strings to check against the truth table in Table~\ref{tab:truth_table_PcPd}, which is computationally expensive. To reduce this complexity, we employ a grading system: we first identify a pair of $P_a$ and $P_b$ with the highest grade, and then select $P_c$ and $P_d$ based on the number of Pauli strings that can be generated.

The grading system is defined as follows. For each generated Pauli string, we first identify its corresponding generators. For instance in Fig.~\ref{fig:NCF}(a), $P_5$ is generated by $P_1$ and $P_3$. For each candidate pair of anti-commuting Pauli strings (i.e., potential $P_a$ and $P_b$), we assign a grade based on their role in generating other Pauli strings: 3 points if the pair can directly generate a Pauli string, and 1 point if one or both Pauli strings appear in the generator set of a generated Pauli string. This scoring prioritizes pairs that contribute most to generating additional Pauli strings, thus more Pauli strings can be grouped.

After selecting the candidate pair with the highest grade, we then consider all possible pairs of candidate $P_c$ and $P_d$ from the remaining generator Pauli strings. If the combination of the chosen $P_a$ and $P_b$ with a candidate $P_c$ and $P_d$ satisfies the truth table in Table~\ref{tab:truth_table_PcPd}, we include the four Pauli strings along with their generated Pauli strings as a candidate group. Among all such candidate groups, we select the one that contains the largest number of Pauli strings. Note that it is possible that none of the candidate pairs $P_c$ and $P_d$ satisfy the truth table. In this case, we select a single $P_c$ such that the three Pauli strings $P_a$ to $P_c$ generate the largest number of Pauli strings. This is because any generator Pauli string, together with the chosen $P_a$ and $P_b$, can be conjugated into three Pauli strings that act non-trivially on two qubits~\cite{van2021simple}. Similar to the single-qubit grouping, we refer to the already-formed groups as ``anticommuting groups,'' and include the remaining ungrouped mutually commuting Pauli strings into ``commuting groups'' to further reduce the T-gate depth.

Note that, with four generator Pauli strings, we can generate at most 15 distinct Pauli strings. 
The reason the maximum number is 15 is as follows: for each qubit, there are four possible Pauli operators ($X$, $Y$, $Z$, and $I$), giving a total of $4^2=16$ combinations for two qubits. Excluding the $II$ operator, which acts trivially on both qubits, we are left with at most 15 distinct Pauli strings.

Using this method, we can partition all six Pauli strings in Fig.~\ref{fig:NCF} into two groups, as illustrated in Fig.~\ref{fig:NCF}(f), where $P_1$, $P_3$, $P_4$, and $P_6$ serve as $P_a$ to $P_d$, respectively. Using the Clifford circuit generation approach in Section~\ref{sec:clifford}, the conjugated transformed Pauli strings are shown in Fig.~\ref{fig:NCF}(g), and the associated quantum circuit, annotated with the Clifford operations, is presented in Fig.~\ref{fig:NCF}(h). As seen in Fig.~\ref{fig:NCF}(h), the transformed Pauli strings $P'_1$, $P'_3$, $P'_4$, $P'_5$, and $P'_6$ act non-trivially on only two qubits (i.e., $q_0$ and $q_2$).


\subsection{Clifford Circuit Generation}
\label{sec:clifford}
In the previous section, we introduced two methods to partition multiple Pauli strings into groups. For each group, a Clifford circuit is applied to conjugate the Pauli strings, enabling them to act non-trivially on one or two qubits. Based on Rule-II in Section~\ref{sec:conju}, the conjugated Pauli strings, along with the reversed Clifford circuit, are implemented to preserve the function of the original Pauli strings. In this section, we describe how to generate the Clifford circuit and the conjugated Pauli strings for a selected group of Pauli strings.

For the generators in a group (i.e., the anticommuting pair or $P_a$ to $P_d$), we represent them using a tableau. As discussed in Section~\ref{sec:conju}, quantum gates can update the tableau column by column. We thus search for a sequence of quantum gates that transforms the tableau into a new form in which only one or two same columns of the X and Z matrices contain 1s, indicating that the Pauli strings act non-trivially on only one or two qubits. The transformed tableau then represents the conjugated Pauli strings. An example of such a transformation is shown in Fig.~\ref{fig:Clifford}.

The tableau transformation is performed row by row, starting with the first row (i.e., the first Pauli string), which is transformed to contain only a single 1 in the row. This process occurs in two stages. First, we eliminate all 1s in the Z matrix using S and H gates: we apply an H gate to swap a column in the X matrix with the corresponding column in the Z matrix when Z=1 and X=0, and we apply an S gate when both X and Z are 1 in that column. After this stage, only 1s remain in the X matrix of the first row. 

Second, we reduce the number of 1s in the X matrix to a single 1 by applying CNOT gates between columns that contain 1s, performing them in parallel to minimize circuit depth (e.g., between $x_0$ and $x_1$ and simultaneously between $x_2$ and $x_3$). When multiple CNOT combinations are possible, we select the one that minimizes the total number of 1s in the tableau, which reduces the number of Clifford gates required in subsequent iterations. After completing these two stages, only a single 1 remains in the first row of the X matrix. We then apply an H gate to transfer this 1 to the Z matrix. The column containing this 1 is referred to as the ``pivot column.''

For the second row, we follow the same procedure: H and S gates are applied to eliminate all 1s in the Z matrix, followed by CNOT gates to reduce the 1s in the X matrix to a single 1 located in the pivot column. For example, in Fig.~\ref{fig:Clifford}, the first row of the Z matrix only contains a single 1, so no operations are needed, and the pivot column is 0. In the second row, we first apply S gates to the first and second qubits to eliminate the 1s in the Z matrix. Next, we apply a CNOT gate to remove the extra 1 in the X matrix. After these operations, both 1s in the first two rows align in the same column (i.e., column 0).

The single-qubit case stops at this point. For the two-qubit case, we follow the same procedure, reducing the 1s in each row to a single column, excluding the pivot column of the first two rows. For example, in Fig.~\ref{fig:Clifford}, after applying two S gates and one CNOT gate, only one 1 remains in the Z matrix in the third row, so we set column 2 as the new pivot column and target the reduction of 1s in the fourth row. After applying an S gate and two CNOT gates, the 1s remain only in column 0, the pivot for the first two rows, and in column 2, the pivot for the last two rows. This procedure yields the Clifford circuit and the corresponding conjugated Pauli strings. As discussed in Section~\ref{sec:grouping}, any additional generated Pauli strings in the group are also conjugated by the same Clifford circuit and act non-trivially on only one or two qubits.

After generating the Clifford circuits and conjugated Pauli strings for the 
anticommuting groups, we generate those for the commuting groups. For the commuting 
groups, we apply the same strategy to transform each row (i.e., Pauli string) in the 
matrices to contain a single 1, each in a different column. In other words, each Pauli string is conjugated to act non-trivially on a different qubit. This allows the parallel execution of RZ gate rotations within the group.

\begin{figure}
    \centering
   \includegraphics[width=1\linewidth]{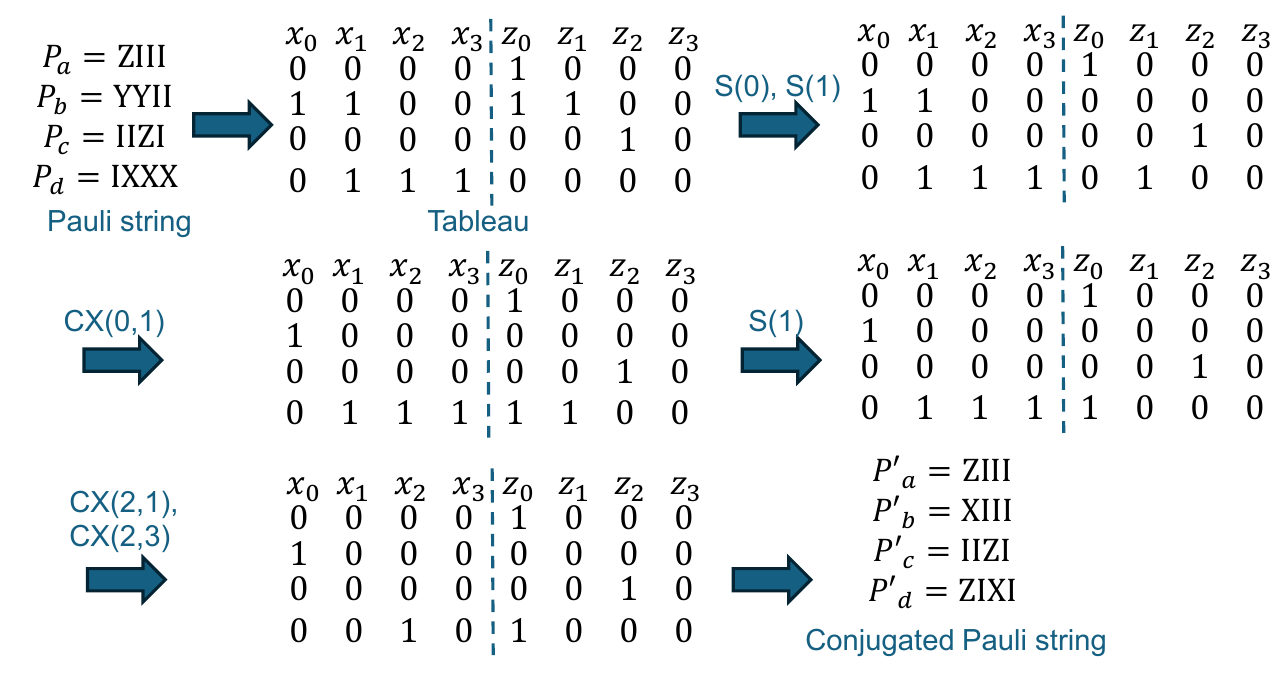}
     \caption{An example of conjugation and Clifford circuit generation.}
    \label{fig:Clifford}
 \end{figure}

\subsection{Sliding Window Strategy}
\label{sec:window}
Although we employ a heuristic algorithm, the complexity of this approach can still be high. For instance, Gaussian elimination has a complexity of $O(m^3)$, where $m$ is the number of Pauli strings. To reduce computational cost, we adopt a sliding window strategy, limiting the algorithm to a subset of $w$ Pauli strings in each iteration. Specifically, from the ungrouped Pauli strings, we first identify the initial pair of anticommuting Pauli strings and include them in the window, since both single- and two-qubit grouping require at least one anticommuting pair. We then select the remaining $w-2$ ungrouped Pauli strings from the beginning of the list to complete the window. We include a sensitivity study in Section~\ref{sec:sensi} to investigate the optimal number of $w$.

\subsection{Complexity Analysis} 
In this section, we analyze the complexity of NCF, starting with the grouping algorithm. Since the two-qubit grouping involves more steps and has a higher complexity, we focus on it. The process begins with generating the anticommuting and commuting graphs. In this step, we compare all pairs of $m$ Pauli strings, which results in $ \binom{m}{2} = \frac{m^2-m}{2} $ comparisons, each involving $q$ qubits. Thus, this step has a complexity of $O(\frac{(m^2-m)q}{2})$. Since each group contains at most 15 Pauli strings, the number of iterations is proportional to $m$. In each iteration, we first perform a Gaussian elimination with a complexity of $O(m^3)$. For each pair of anticommuting generator Pauli strings, we conduct a grading process. Given that the maximum number of generators is $2q$~\cite{aaronson2004improved}, there are at most $ \binom{2q}{2} = 2q^2-q $ such pairs. During the grading, we check if a pair can generate or is part of the generator set of a generated Pauli string. Since the number of generated Pauli strings is at most $m$, the complexity of the grading process is $O(m(2q^2 - q))$. After selecting the highest-point anti-commuting pair, we consider all possible pairs of candidate Pauli strings $P_c$ and $P_d$ from the remaining generators, which amounts to $\binom{2q}{2} = 2q^2 - q$ pairs. For each pair, we then determine how many Pauli strings they can generate among the at most $m$ candidates. Finally, we reorder each group so that it can be executed in parallel with the subsequent groups, which has a complexity of $O(m \log m)$. The total complexity of the grouping approach is therefore given by $O(\frac{(m^2-m)q}{2}+m(m^3+2mq^2-mq+2mq^2-mq +m\log m))\approx O(m^4+4m^2q^2)$. As discussed in Section~\ref{sec:window}, we use a sliding window strategy to limit the complexity by only considering the Pauli strings within a window of size $w$. This reduces the complexity of the grouping algorithm to $O(m(w^3 + 4wq^2))$, as the maximum number of Pauli strings is limited to $w$.

For Clifford circuit generation, we repeatedly generate single-qubit gates and CNOT gates for each group of Pauli strings, where at most $q$ single-qubit gates and $q$ CNOT gates are needed. During the CNOT gate generation, we compare all possible pairs of CNOT gates, which results in at most $\binom{q}{2} = \frac{q^2-q}{2}$ pairs, and select the ones that minimize the number of 1s in the tableau. Therefore, the complexity of this stage is $O(m(q+\frac{q^2-q}{2}))=O(\frac{mq^2+mq}{2})$.

\section{Evaluation}
\label{sec: evaluate}
\subsection{Experiment Setup}
\label{sec:setting}
\subsubsection{Benchmark} \label{sec:benchmark}
We perform Hamiltonian simulation for five molecules (LiH, H2O, N2, H2S, and CO2) using PySCF~\cite{sun2018pyscf}, where each molecular Hamiltonian involves a different number of qubits and Pauli strings. In addition, we generate the Pauli strings for both the Ising and Heisenberg models, which are widely used in physics research~\cite{georgescu2014quantum}. For these models, we consider two different lattice structures (2D and 3D) and two qubit counts (30 and 60). The detailed information, including the logical qubit count and the number of Pauli strings, are shown in Table~\ref{tab:benchmarks}.

\begin{table}[h!]
\centering
\caption{Qubit count and number of Pauli strings for the simulated Hamiltonians.}
\label{tab:benchmarks}
\begin{tabular}{|l|c|c|c|}
\hline
\textbf{Type} & \textbf{Structure} & \textbf{Qubit Count} & \textbf{Pauli Strings Count} \\
\hline
& LiH  & 12 & 630 \\
\cline{2-4}
& H2O & 14 & 1085 \\
\cline{2-4}
Molecule & N2 & 20 & 2950 \\
\cline{2-4}
 & H2S & 22 & 6245 \\
\cline{2-4}
& CO2 & 30 & 16121 \\
\hline
 & 2D & 30 & 79 \\
 \cline{2-4}
Ising & 2D & 60 & 164 \\
 \cline{2-4}
 & 3D & 30 & 89 \\
  \cline{2-4}
 & 3D & 60 & 193 \\
\hline
 & 2D & 30 & 147 \\
  \cline{2-4}
Heisenberg & 2D & 60 & 312 \\
 \cline{2-4}
 & 3D & 30 & 177 \\
  \cline{2-4}
 & 3D & 60 & 399 \\
\hline
\end{tabular}
\end{table}

\subsubsection{Baseline} \label{sec:baseline}
We use two baselines to demonstrate the effectiveness of NCF: (i) Gridsyn~\cite{ross2014optimal} and (ii) Rustiq+Trasyn~\cite{debrugière2024faster, hao2025reducing}. Gridsyn is a widely used synthesizer for RZ gates. As explained in Section~\ref{sec:motivation}, the number of RZ gates in a Hamiltonian simulation equals the number of Pauli strings. Therefore, we apply Gridsyn to synthesize each RZ gate. In contrast, Rustiq can merge certain RZ gates into U3 gates, after which we use Trasyn to synthesize the resulting unitaries.

\subsubsection{Metrics}
We use three metrics to evaluate NCF and the baselines: \textbf{T-gate count}, \textbf{T-gate depth}, and \textbf{Clifford count}. Although the implementation cost of Clifford gates is significantly lower than that of T gates, it remains non-negligible. Thus, we include the Clifford count as one metric.

\begin{figure*}
    \centering
   \includegraphics[width=1\linewidth]{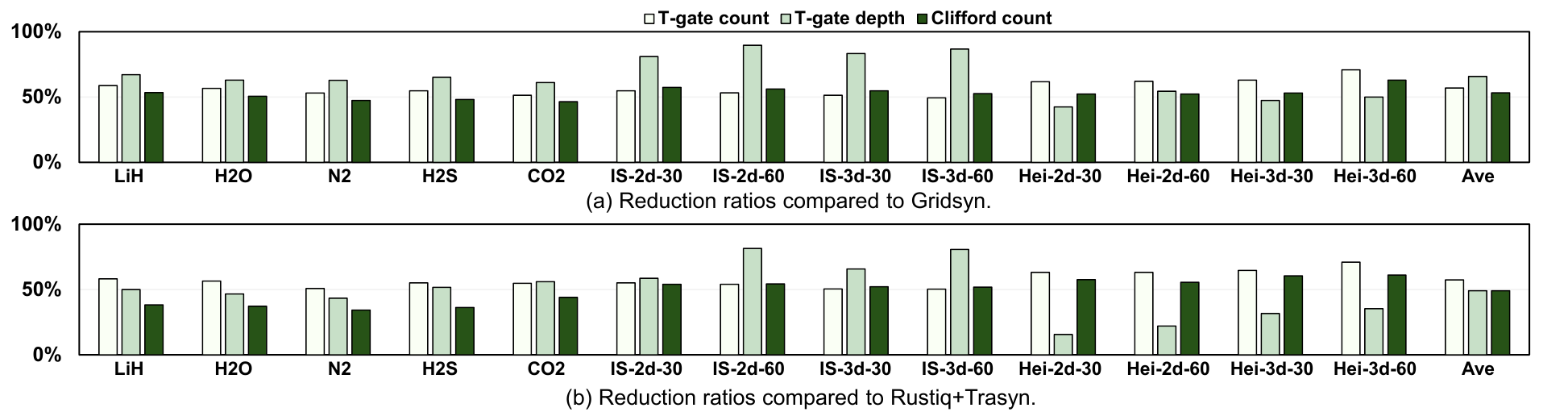}
     \caption{Reduction ratios achieved by single-qubit NCF.}
    \label{fig:single_main}
 \end{figure*}

   \begin{table*}[t]
\centering
\footnotesize
\setlength\tabcolsep{4.8pt}
\caption{Comparison of four methods across benchmarks in terms of T-gate count, T-gate depth, and Clifford count.}
\label{tab:two-qubit}
\begin{tabular}{|l||c|c|c||c|c|c||c|c|c||c|c|c|}
\hline
\multirow{2}{*}{\textbf{Benchmark}} & 
\multicolumn{3}{c||}{\textbf{Gridsyn}} & 
\multicolumn{3}{c||}{\textbf{Rustiq+Trasyn}} & 
\multicolumn{3}{c||}{\textbf{Single-qubit NCF + Trasyn}} & 
\multicolumn{3}{c|}{\textbf{Two-qubit NCF + Synthetiq}} \\ \cline{2-13}
 & \textbf{T-count} & \textbf{T-depth} & \textbf{Clifford} 
 & \textbf{T-count} & \textbf{T-depth} & \textbf{Clifford}
 & \textbf{T-count} & \textbf{T-depth} & \textbf{Clifford}
 & \textbf{T-count} & \textbf{T-depth} & \textbf{Clifford} \\ \hline

LiH & 1,419 & 1,352 & 12,774 & 1,397 & 757 & 8,890 & 818 & 716 & 5,767 & \textbf{49} & \textbf{42} & \textbf{5,289} \\ 
H2O & 3,746 & 3,415 & 27,017 & 3,662 & 2,412 & 14,111 & 1,262 & 1,035 & 14,434 & \textbf{375} & \textbf{303} & \textbf{11,691} \\ 
N2 & 7,537 & 7,431 & 82,196 & 7,138 & 4,237 & 5,2431 & 3,578 & 3,101 & 45,955 & \textbf{541} & \textbf{516} & \textbf{41,575} \\ 
Ising-2D-30 & 790 & 90 & 1,899 & 790 &  90 & 1,270 & 223 & 20 & 593 & \textbf{97} & \textbf{16} & \textbf{533} \\ 
Ising-2D-60 & 1,640 & 290 & 3,924 & 1,640 & 160 & 2,664 & 488 & 20 & 1,363 & \textbf{200} & \textbf{16} & \textbf{1,100} \\ 
Ising-3D-30 & 890 & 220 & 2,109 & 890 & 90 & 1,765 & 293 & 29 & 680 & \textbf{90} & \textbf{21} & \textbf{595} \\ 
Ising-2D-60 & 1,930 & 280 & 4,533 & 1,930 & 190 & 3,234 & 691 & 26 & 1,498 & \textbf{195} & \textbf{20} & \textbf{1,327} \\ 
Heisenberg-2D-30 & 2,352 & 864 & 4,570 & 2,352 & 556 & 4,531 & 521 & 87 & 1,222 & \textbf{224} & \textbf{52} & \textbf{1,117} \\ 
Heisenberg-2D-60 & 4,993 & 1,345 & 9,725 & 4,972 & 736 & 9,665 & 1,131 & 94 & 2,679 & \textbf{493} & \textbf{80} & \textbf{2,479} \\ 
Heisenberg-3D-30 & 2,832 & 1,008 & 5,831 & 2,832 & 720 & 5,465 & 652 & 108 & 1,523 & \textbf{253} & \textbf{68} & \textbf{1,378} \\ 
Heisenberg-3D-60 & 6,384 & 1,872 & 16,449 & 6,382 & 1,296 & 12,344 & 1,481 & 115 & 3,472 & \textbf{640} & \textbf{107} & \textbf{3,234} \\ 
\hline
\rowcolor{gray!30}
Average Reduction & \textbf{90.6\%} & \textbf{92.5\%} & \textbf{68.7\%} & \textbf{90.5\%} & \textbf{88.3\%} & \textbf{53.3\%} & \textbf{66.9\%} & \textbf{38.1\%} & \textbf{13.4\%} & & &  \\ 
\rowcolor{gray!30}
Ratios &  &  &  &  &  &  &  &  &  &  &  &  \\ 

\hline
\end{tabular}
\end{table*}

\subsubsection{Implementations}
\label{sec:implement}
In the baseline, we use an error threshold of $\epsilon = 0.001$ for Gridsyn, which has been shown to be sufficient for Hamiltonian Simulation~\cite{hao2025reducing, fomichev2025fast}. For the Rustiq+Trasyn baseline and for single-qubit NCF, we scale the error threshold for each synthesis to ensure that the total logical error rate remains consistent across different settings. Specifically, we set the adjusted threshold as $\epsilon=0.001\frac{Num\_Paulis}{Num\_unitaries}$, where $Num\_unitaries$ denotes the number of unitaries after applying Rustiq or NCF, and $Num\_Paulis$ is the original number of Pauli strings.

In the two-qubit grouping case, we employ Synthetiq as the two-qubit unitary synthesizer~\cite{paradis2024synthetiq}. We observe that Synthetiq can efficiently synthesize arbitrary two-qubit unitaries when the error threshold is above $\epsilon = 0.12$, whereas synthesizing a unitary with a lower threshold can take more than six hours. Accordingly, in our two-qubit grouping evaluation, we fix the error threshold for each Synthetiq synthesis at $\epsilon = 0.12$ and proportionally scale the error thresholds used in the baselines (Gridsyn and Rustiq+Trasyn). 

In terms of the NCF setting, we set the window size $w$ for the single-qubit case to 4 and the two-qubit case to 128. 
A more detailed analysis of different choices for $w$ will be presented in the sensitivity study in Section~\ref{sec:window}.

\subsection{Single-qubit Results} 
\label{sec:result_single}
In this section, we present the results of single-qubit NCF compared to the baselines. As shown in Fig.~\ref{fig:single_main}(a), NCF achieves average reductions of 57.0\%, 65.4\%, and 52.9\% in T-gate count, T-gate depth, and Clifford count, respectively, relative to Gridsyn. Compared to Rustiq+Trasyn, the reductions are 57.4\%, 49.1\%, and 49.0\% in T-gate count, T-gate depth, and Clifford count, respectively. NCF consistently outperforms both baselines across all 13 benchmarks listed in Table~\ref{tab:benchmarks}, achieving significant improvements in all three metrics.

As shown in the results, NCF achieves nearly a 60\% improvement in T-gate count compared to both baselines. This improvement arises because almost all Pauli strings can be partitioned into anticommuting groups, with each group containing at least two Pauli strings. Such grouping enables the merging of at least two unitaries into a single U3 unitary, effectively reducing the total number of unitaries by approximately 50\%. Moreover, the reduced number of unitaries allows Trasyn to operate with higher synthesis precision, which further decreases the number of synthesized T gates. We observe that the T-gate count reduction ratios are similar when compared to both baselines. Although Rustiq enables limited merging of unitaries, the number of such merges is negligible and does not significantly impact the overall reduction~\cite{debrugière2024faster}.  

In terms of T-gate depth, NCF achieves more than a 60\% reduction compared to Gridsyn. This improvement results from the anti-commuting grouping, commuting grouping, and the reordering strategy in NCF, which not only reduces the number of T gates but also increases their parallelism. When compared to Rustiq+Trasyn, the reduction is slightly below 50\%, as Rustiq optimizes the parallel execution of single-qubit unitaries~\cite{debrugière2024faster}.
Interestingly, NCF also achieves substantial improvements in Clifford count compared to both baselines. This arises from two main factors: (i) the number of synthesized Clifford gates is proportional to the number of synthesized T gates~\cite{hao2025reducing, ross2014optimal}, so reducing the T-gate count through unitary merging also lowers the Clifford count, and (ii) as discussed in \cite{debrugière2024faster}, simultaneous conjugation of multiple Pauli strings can further decrease the number of Clifford gates. For instance, the Clifford count is reduced from 24 to 20 when comparing the circuit in Fig.\ref{fig:NCF}(b) with Fig.~\ref{fig:NCF}(e). The reduction ratio against Rustiq+Trasyn, however, is smaller than that against Gridsyn, since Rustiq primarily targets reducing the CNOT gate count.

\subsection{Two-qubit Results} 
\label{sec:two-result}
As discussed in Section~\ref{sec:implement}, we fix the error threshold for two-qubit Synthetiq synthesis and scale the thresholds for both baselines accordingly. To determine the best compilation strategy, we also include single-qubit NCF for comparison in this section. For fairness, the error threshold for Trasyn synthesis in the single-qubit NCF is scaled in the same way as in the baselines. As a result, we evaluate four methods: two-qubit NCF, single-qubit NCF, Rustiq+Trasyn, and Gridsyn. The total error is kept constant across all methods, defined as the product of the per-synthesis error threshold and the number of unitaries. 
Due to the long execution time of Synthetiq, we exclude the H2S and CO2 benchmarks in this experiment. The T-gate count, T-gate depth, and Clifford count achieved by the four methods are summarized in Table~\ref{tab:two-qubit}.  We also provide the averaged reduction ratios achieved by the two-qubit NCF compared to the other three methods for each metric in Table~\ref{tab:two-qubit}. 

Based on Table~\ref{tab:two-qubit}, we have the following four observations: (i) As observed, the two-qubit NCF consistently achieves the best performance among the four methods, while the single-qubit NCF ranks second.  
(ii) Across different benchmarks, the reduction ratios of the two-qubit NCF are particularly higher for LiH, H2O, and N2. This is because the Pauli strings for these molecules are more complex than those in the Ising and Heisenberg models, allowing the two-qubit NCF to explore more grouping opportunities.
 (iii) Compared to Gridsyn, Rustiq+Trasyn achieves a similar T-gate count while providing improvements in T-gate depth and Clifford count across most benchmarks. These gains stem from Rustiq’s ability to merge a limited number of unitaries, enable parallel execution of RZ gates, and reduce Clifford gates such as CNOTs. This observation is consistent with the results presented in Section~\ref{sec:result_single}. (iv) Compared to single-qubit NCF, two-qubit NCF achieves a substantial reduction in T-gate count, while the improvements in T-gate depth and Clifford count are more modest. This is because multi-qubit grouping merges more unitaries into a single unitary, significantly lowering the number of T gates. However, a two-qubit group can include up to 15 Pauli strings (compared to 3 in the single-qubit case) and therefore acts non-trivially on more qubits, which reduces the potential for parallel execution and limits the reduction in T-gate depth. Since each group applies quantum gates to all qubits involved in its Pauli strings, the resulting unitaries are wider. For example, in Fig.~\ref{fig:NCF}(e), the two one-qubit groups act on two and three qubits, respectively, allowing them to be executed in parallel with other groups. In contrast, in Fig.~\ref{fig:NCF}(h), the two-qubit group acts on all four qubits, which prevents parallel execution with other groups. Moreover, the number of Clifford gates in the Clifford circuit remains similar between single-qubit and two-qubit grouping, resulting in comparable overall Clifford counts. For instance, the Clifford count in the Clifford circuit is similar between Fig.\ref{fig:NCF}(e) and Fig.\ref{fig:NCF}(h).   

Although two-qubit grouping demonstrates superior results compared to other methods, we believe single-qubit grouping remains the most practical strategy at the current stage. This is because existing two-qubit synthesizers are limited in scalability and only support relatively high error thresholds (e.g., $\epsilon = 0.12$ in Synthetiq), which are impractical for fault-tolerant quantum applications. In contrast, single-qubit synthesizers such as Trasyn can achieve much tighter error thresholds (below $\epsilon = 0.001$), enabling their use in practical implementations~\cite{hao2025reducing, fomichev2025fast}. However, with the development of more advanced multi-qubit synthesizers, two-qubit grouping has the potential to become the superior strategy, as it achieves the best performance across all metrics.

\begin{figure}
    \centering
   \includegraphics[width=1\linewidth]{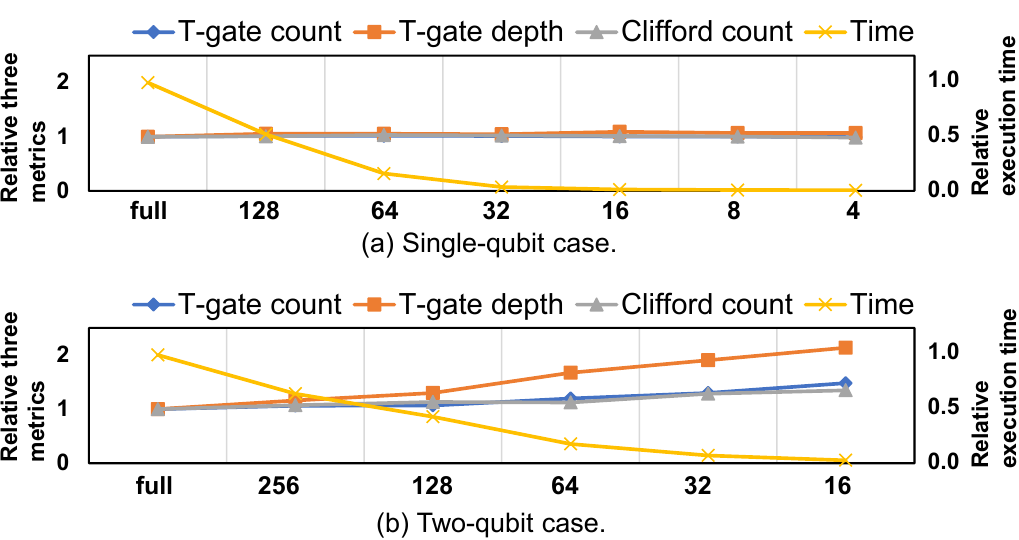}
     \caption{Relative values compared to the full window size.}
    \label{fig:sensitivity}
 \end{figure}

\subsection{Sensitivity Study of Window Size} 
\label{sec:sensi}
In this section, we conduct a sensitivity study to determine the optimal window size $w$ for both single-qubit and two-qubit NCF. For the single-qubit case, we evaluate seven window sizes: full size (equal to the total number of Pauli strings), 128, 64, 32, 16, 8, and 4. For the two-qubit case, we evaluate six window sizes: full size, 256, 128, 64, 32, and 16. The minimum window size is set to 4 and 16 for single- and two-qubit NCF, respectively, since these values are larger than the maximum possible group size (3 for single-qubit and 15 for two-qubit). Choosing a smaller window size would prevent NCF from grouping all eligible Pauli strings.

We selected four benchmarks for this experiment: LiH, H2O, Ising-2D-60, and Heisenberg-2D-60. Each benchmark is executed under the seven window sizes in the single-qubit case and the six window sizes in the two-qubit case. To ensure a fair comparison, we fix the total error rate for each benchmark as described in Section~\ref{sec:two-result}. For each window size, we measure the three metrics along with the compilation time and report the normalized values relative to the full-size window in Fig.~\ref{fig:sensitivity}.

As one can see, the three metrics remain nearly constant in the single-qubit case, while the compilation time decreases as the window size reduce. This is because most one-qubit groups contain only two generator Pauli strings without their generated Pauli string, since a generated Pauli string is typically produced by a larger number of generators. Consequently, increasing the window size does not capture additional Pauli strings within a group and therefore provides no further reduction in the three metrics. Based on these observations, we set the window size to 4 for the single-qubit case.

For the two-qubit case, we observe that all three metrics increase as the window size decreases. This is because a larger window provides NCF with a higher chance of locating the generated Pauli strings, thereby reducing the number of unitaries. To balance metric performance and compilation time, we set the window size to 128 for the two-qubit case.

\section{Related Works and discussion}
Several compilation frameworks have been proposed for Hamiltonian simulation~\cite{berg2020circuit, debrugière2024faster, li2022paulihedral, liu2025quclear}. However, these methods primarily focus on reducing the number of CNOT gates, which is not directly applicable to fault-tolerant quantum computers. To the best of our knowledge, this work is the first compilation framework specifically targeting Hamiltonian simulation for fault-tolerant quantum computers. Our framework can also be applied to future Hamiltonian simulations that utilize quantum phase estimation (QPE), in which each $RZ$ rotation gate is replaced by a controlled-$RZ$ gate~\cite{georgescu2014quantum, fomichev2025fast}. Since each controlled-$RZ$ gate can be transformed into a $ZZ$ gate~\cite{fomichev2025fast}, each Pauli string in the QPE-based Hamiltonian simulation acquires an additional $Z$ operator on the extra qubit, making it compatible with our framework.

Our framework relies on single-qubit U3 gate synthesizers and multi-qubit synthesizers to decompose the merged unitaries into sequences of Clifford and T gates. The development of unitary synthesizers is an active area of research, including dedicated RZ gate synthesizers~\cite{ross2014optimal, paetznick2013repeat}, U3 gate synthesizers~\cite{paetznick2013repeat, kliuchnikov2022shorter, hao2025reducing}, and multi-qubit synthesizers~\cite{paradis2024synthetiq, kliuchnikov2022shorter, gheorghiu2022t}. Although current state-of-the-art U3 synthesizers are limited to relatively high error thresholds (e.g., 0.001 for Trasyn) compared to RZ gate synthesizers~\cite{ross2014optimal, paetznick2013repeat, hao2025reducing}, and multi-qubit synthesizers are constrained by slower execution times~\cite{paradis2024synthetiq, gheorghiu2022t}, these methods have been shown to have the potential to achieve a similar T-gate count under the same error threshold~\cite{kliuchnikov2022shorter}. With advances in synthesizer methods, both U3 and multi-qubit synthesizers can achieve lower error thresholds and faster execution times, making our approach more practical for reducing T-gate count and depth in Hamiltonian simulation.


\section{Conclusion}
In this paper, we propose NCF, a compilation framework aimed at reducing both the T-gate count and T-gate depth for Hamiltonian simulation. For a list of Pauli strings in the Hamiltonian, NCF partitions them into groups and conjugates the Pauli strings within each group. After conjugation, the Pauli strings in each group act non-trivially on only one or two qubits, enabling the simultaneous synthesis of multiple $R_Z$ gates using a U3 or multi-qubit block synthesizer. Experimental results demonstrate that NCF achieves average reductions of 57.4\%, 49.1\%, and 49.0\% in T-gate count, T-gate depth, and Clifford count, respectively, compared to the state-of-the-art method.

\section*{Acknowledgment}

This material is based upon work supported by the U.S. Department of Energy, Office of Science, National Quantum Information Science Research Centers. This material is also based upon work supported by the DOE-SC Office of Advanced Scientific Computing Research MACH-Q project under contract number DE-AC02-06CH11357.



\bibliographystyle{IEEEtran}
\bibliography{refs}

\end{document}